\documentclass{patmorin}
\usepackage{amsthm,amsmath,graphicx,stmaryrd}
\usepackage{pat}
\DeclareMathOperator{\erf}{erf}
\DeclareMathOperator{\area}{area}

\title{\MakeUppercase{On the Average Number of Edges in Theta Graphs}}
\author{Pat Morin and Sander Verdonschot%
	\thanks{School of Computer Science, Carleton University}}

\begin{document}
\begin{titlepage}
\maketitle

\begin{abstract}
  Theta graphs are important geometric graphs that have many applications,
  including wireless networking, motion planning, real-time animation, and
  minimum-spanning tree construction.  We give closed form expressions
  for the average degree of theta graphs of a homogeneous Poisson
  point process over the plane.  We then show that essentially the same
  bounds---with vanishing error terms---hold for theta graphs of finite
  sets of points that are uniformly distributed in a square.  Finally,
  we show that the number of edges in a theta graph of points uniformly
  distributed in a square is concentrated around its expected value.
\end{abstract}
\end{titlepage}

\section{Introduction}

Theta graphs
\cite{clarkson:approximation,keil:approximating,keil.gutwin:classes}
are important geometric graphs that have many applications,
including wireless networking \cite{alzoubi.li.ea:geometric},
motion planning \cite{clarkson:approximation}, real-time animation
\cite{fischer.lukovszki.ea:geometric}, and minimum-spanning tree
construction \cite{yao:on}.  These graphs are defined on a planar point
set, $S$, with an integer parameter $k$.  For each $i\in\{1,\ldots,
k\}$, define the cone
\[ 
   C_i=\{u\in \R^2: \measuredangle qou\in [2\pi(i-1)/k,2\pi i/k)\} \enspace ,
\]
where $q=(1,0)$ and $o=(0,0)$.  In the $\theta_k$-graph, $\theta_k(S)$,
each point $u\in S$ has an edge connecting it to the nearest point,
if any, in the cone $C_i+u$, for each $i\in\{1,\ldots,k\}$.  Here,
``nearest'' has a special meaning: The theta graph connects $u$ to the
point whose orthogonal projection on the axis of $C_i+u$ is closest to $u$
(see \figref{theta-vs-yao}).

\begin{figure}
  \begin{center}
     \begin{tabular}{c@{\hspace{1in}}c}
     \includegraphics{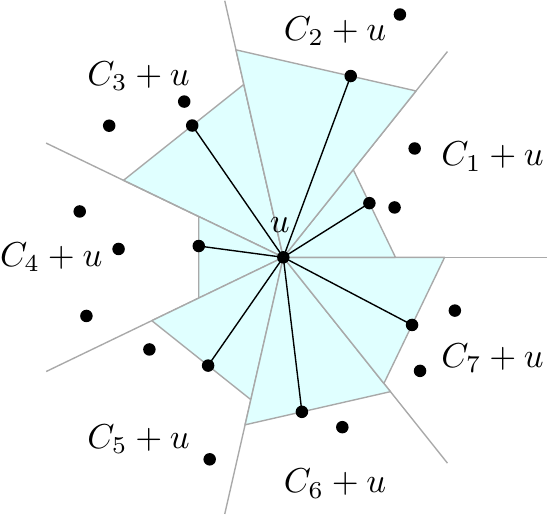} &
     \includegraphics{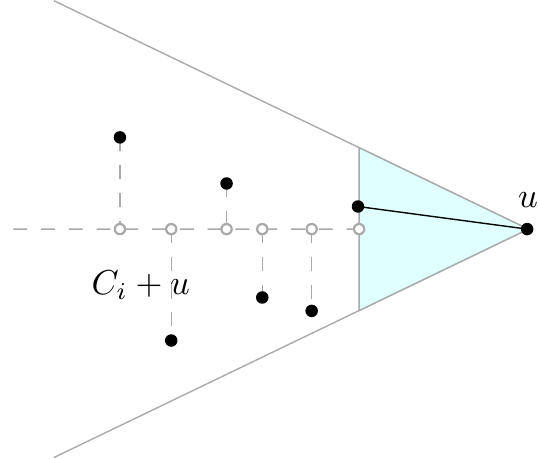} \\
     (a) & (b)
     \end{tabular}
  \end{center}
  \caption{A vertex, $u$, in a $\theta_7$-graph; (a)~$u$ connects to the
  ``nearest'' point in each cone; (b)~``nearest'' is defined in terms of 
  the projection onto the axis of each cone.}
  \figlabel{theta-vs-yao}
\end{figure}

Theta graphs have two important properties that make them suited to a
wide variety of applications:  They are \emph{sparse}; $\theta_k(S)$
has at most $k|S|$ edges and they are \emph{spanners}; the length
of the shortest path between any two vertices $u$ and $w$ is at most
a constant (depending only on $k$ and not on $S$) times the Euclidean
distance between $u$ and $w$.  For any point set, $S$, $\theta_k(S)$ is
a spanner for any $k\ge 4$
\cite{barba.bose.ea:on,bonichon.gavoille.ea:connections,%
bose.morin.ea:theta-5,keil:approximating,ruppert.seidel:approximating}.

Note that, although $\theta_k(S)$ has at most $k|S|$ edges, it can also
have significantly fewer edges.  For example, if the points of $S$ all
lie on a line, then $\theta_k(S)$ has only $|S|-1$ edges.  More typical,
though, is for $\theta_k(S)$ to have somewhere between $k|S|/2$ and $k|S|$
edges;  each vertex $u\in S$ chooses $k$ edges\footnote{Specifically only
the points that are \emph{unoriented $\theta$-maxima} of $S$ may choose
fewer than $k$ edges; the number of unoriented $\theta$-maxima is typically
$O(\sqrt{|S|})$ \cite[Theorem~4]{avis.beresford-smith.ea:unoriented}.}
of the graph but sometimes an edge $uw$ is chosen both by $u$ and $w$
and therefore it should not be counted twice.

\subsection{The Models and Results}

In this paper, we study the typical number of edges in $\theta_k$-graphs
by studying the average number of edges in two different models of random
point sets.

We begin, in \secref{poisson}, by studying (the infinite)
$\theta_k$-graphs generated by a homogeneous Poisson point process
with unit intensity over the entire Euclidean plane.  In this model,
we study the average degree of a vertex in the $\theta_k$-graph.
For the $\theta_k$-graph, this quantity is at most $2k$ since each
vertex defines $k$ edges of the graph and each edge has 2 endpoints.
However, in some cases an edge $uw$ is \emph{mutual} in the sense that
the edge is created both by $u$ and by $w$.  If we let $p_k$ denote the
probability that an edge of the $\theta_k$-graph is mutual, then the
average degree of the $\theta_k$ graph is
\[
    d_k = (2-p_k)k \enspace .
\]
(The second term corrects for the double-counting of mutual edges.)  Thus,
understanding the average degree of a $\theta_k$-graph boils down to
computing $p_k$.

In \secref{even} we show that, for all even integers $k\ge 4$,
\[
    p_k=\frac{\pi\sqrt{3}}{9}\approx 0.6045997883 \enspace .
\]
That is, the probability of an edge being mutual is independent of
$k$. Thus, for all even integers $k\ge 4$, the average degree of a
vertex is
\[
  d_k = \left(2-\frac{\pi\sqrt{3}}{9}\right)k \approx 1.395400212\cdot k \enspace .
\]

In \secref{odd} we show that, for odd values of $k\ge 5$, the situation
is very different.  The mutual edge probability, $p_k$, depends on
$k$ in a complicated way that includes trigonometric functions and
square roots.  However, the value of $p_k$ is significantly larger than
$\frac{\pi\sqrt{3}}{9}$ for all odd values of $k$.  Indeed, $p_k$ is a
decreasing function of $k$ and
\[
  p_k\ge \lim_{k\to\infty} p_k = 2\arctan(1/3)\approx 0.6435011088 \enspace .
\]
Thus, for all odd values of $k\ge 5$,
\[
   d_k \le (2-2\arctan(1/3))k \approx 1.356498891 k
\]
Thus, in some sense, odd values of $k$ offer ``more bang for the buck.''

In \secref{iud}, we also study the \emph{i.u.d.\ model}, in which a
set, $S$, of $n$ points is independently and uniformly distributed in
a square.  In this model, we show that essentially the same bounds hold.
Specifically, If $m_k$ is the number of edges of $\theta_k(S)$, then
\[
    \E[m_k] \in nd_k/2\pm O(k\sqrt{n\log n}) \enspace \enspace ,
\]
where $d_k$, defined above, is the average degree of the $\theta_k$-graph
in the Poisson model.  We also give a concentration result that shows
that the number of edges, $m_k$, is highly concentrated around its
expected value.  In particular
\[
    \Pr\{|m_k - nd_k/2| \ge k\sqrt{cn\log n}\} \le n^{-\Omega(c)} \enspace .
\]

\subsection{Related Work}

As discussed in the introduction, a plethora of literature exists on
theta graphs and their applications, though most of this work focuses on
worst-case analysis.  One notable exception is the work of Devroye \etal\
\cite{devroye.gudmundsson.ea:on} who study the maximum degree of theta
graphs and show that, if $S$ is a set of $n$ points independently and
uniformly distributed in a certain unit square, then $\theta_k(S)$
has maximum degree concentrated around $\Theta((\log n)/\log\log
n)$.\footnote{Devroye \etal\ actually consider the closely-related Yao graphs
\cite{flinchbaugh.jones:strong,yao:on}, but their proofs apply,
almost without modification, to theta graphs.}

In contrast, properties of other proximity graphs of random point
sets have been studied extensively:  
\begin{itemize}
\item Devroye \cite{devroye:expected}
presents a general theorem for obtaining exact leading constants for the
expected degree of a number of proximity graphs over point sets drawn
from a large class of distributions.  This theorem can be applied to
Gabriel graphs, relative neighbourhood graphs, and nearest-neighbour
graphs.  This work \cite[Section~7]{devroye:expected} also mentions
``directional nearest-neighbour graphs,'' now commonly known as Yao graphs
\cite{flinchbaugh.jones:strong,yao:on}, and points out that the general
theorem does not apply to these (nor does it apply to theta graphs---for
the same reasons).

\item 
Penrose and Yukich \cite{penrose.yukich:central,penrose.yukich:weak}
develop weak laws of large number and central limit theorems for several
statistics of proximity graphs of random point sets under some assumptions
about the locality of the graph and the statistic. Their results apply
to statistics like total edge length and number of components of graphs
such as $k$-nearest neighbour graphs, sphere of influence graphs,
and Delaunay triangulations.



\item 
Bern \etal\ \cite{bern.eppstein.ea:expected} study the maximum degree
of Delaunay triangulations of random point sets.  Devroye \etal\ study
the maximum degree of Gabriel graphs \cite{devroye.gudmundsson.ea:on}
of random point sets.  Arkin \etal\ \cite{arkin.anta.ea:probabilistic}
study the length of the longest edge in Delaunay triangulations of random
point sets.

\item 
The issues of connectivity and giant components in the $r$-disk graph
of random point sets---in which an edge $uw$ is present if and only
if the distance between $u$ and $w$ is at most $r$---is the subject of
intensive research and there are at least two books devoted to the topic
\cite{meester.roy:continuum,penrose:random}.
\end{itemize}

\section{The Poisson Model}
\seclabel{poisson}

In the Poisson model, the number of points in any region with area is $A$
follows a Poisson distribution with parameter $A$.  For definitions of
point Poisson processes and distributions see, for example, Daley and
Vere-Jones \cite[Chapter~2]{daley.vere-jones:introduction}.  For our
purposes, the most important properties of the Poisson process are
the following:
\begin{enumerate}
\item The probability that a particular Lebesgue-measurable set, $X$,
   whose area (Lebesgue measure) is $A$ is empty of points is exactly
   $e^{-A}$.
\item For disjoint regions $X_1,\ldots,X_v$, the events ``$X_i$ is empty
   of points'', for $i\in\{1,\ldots,v\}$ are independent.
\end{enumerate}

Throughout this section, and in particular in \secref{odd}, we have made
extensive use of Mathematica to do symbolic integration and manipulation
of trigonometric functions.  Mathematica notebooks containing the code
for these calculations are available at the first author's webpage.

\subsection{Analysis of $p_k$ for Even $k$}
\seclabel{even}

In this section we determine the value of $p_k$ for even values of $k$.
Somewhat surprisingly, the value of $p_k$ in this case does not depend on $k$.

\begin{lem}\lemlabel{even}
 For even integers $k\ge 4$, $p_k=\frac{\pi\sqrt{3}}{9}\approx 0.6045997883$.
\end{lem}

\begin{proof}
  Let $u$ be an arbitrary vertex in a $\theta_k$-graph and let $w$
  be a vertex that $u$ has chosen as a neighbour in one of its cones,
  $C$ ($w$ is the ``closest'' vertex to $u$ in $C$).  Let $T$ be the
  open isosceles triangle defined by $C$ and a line through $w$ that is
  orthogonal to the axis of $C$; refer to \figref{mutual}.  If the edge
  of $T$ opposite $u$ has length $\ell$, then $w$ partitions this edge
  into two pieces of length $r$ and $\ell-r$.

  \begin{figure}
    \centering{\includegraphics{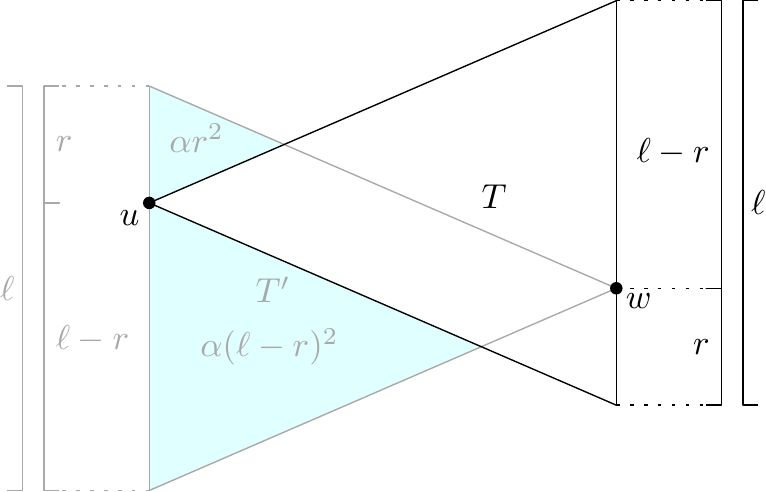}}
    \caption{The edge $uw$ is mutual if and only if $T'\setminus T$ 
       is empty of points.}
    \figlabel{mutual}
  \end{figure}
 
  Let $T'$ be the triangle obtained by reflecting $T$ through the
  midpoint of the edge $uw$ (so that $w$ is a vertex of $T'$). Note
  that, because $k$ is even, $T'$ coincides with one of $u$'s cones.
  In particular, if $w\in u+C_i$, then $T'$ is contained in $w+C_{j}$
  where $j=((i+k/2-1)\bmod k)+1$.  Therefore, the edge $uw$ is mutual if and
  only if $T'\setminus T$ contains no points.  The area of $T'\setminus
  T$ is
  \[
     A(T'\setminus T) = \alpha(r^2+(\ell-r)^2)  \enspace ,
  \]
  where $\alpha=\cos(\theta/2)/(4\sin(\theta/2))$.  We now have
  enough information to compute the probability that the edge $uw$
  is mutual conditional on $\ell$ and $r$:
  \[
    \Pr\{\mbox{$uw$ is mutual} \mid \ell,r\} = \exp(-\alpha(r^2+(\ell-r)^2))
      \enspace .
  \]
  Next, observe that the location of $w$ is uniformly distributed on the
  edge of $T$ opposite $u$, so the value of $r$ (conditioned on $\ell$)
  is uniform over $[0,\ell]$, unconditioning $r$ gives
  \begin{align*}
    f(\ell) & \equiv \Pr\{\mbox{$uw$ is mutual} \mid \ell\} \\
     & = \int_0^\ell 
          (1/\ell)\Pr\{\mbox{$uw$ is mutual} \mid \ell,r\}\,\mathrm{d}r \\
     & = \int_0^\ell (1/\ell)\exp(-\alpha(r^2+(\ell-r)^2))\,\mathrm{d}r \\
     & = \frac{\sqrt{\pi}}{\ell\sqrt{2\alpha}}
            \cdot\exp(-\alpha\ell^2/2)
            \cdot\erf(\ell\sqrt{\alpha/2})  \enspace ,
  \end{align*}
  where 
  \[ \erf(x)=\frac{2}{\sqrt{\pi}}\int_0^x e^{-z^2}\,\mathrm{d}z \]
  is the \emph{Gauss error function}
  \cite[Section~7.2]{abramowitz.stegun:handbook}.

  Next, we remove the conditioning on $\ell$.  The triangle $T$ defines a
  region of area $\alpha\ell^2$ that is empty of points.  Therefore,
  by Property~1 of the Poisson process, the cumulative distribution
  function of $\ell$ is given by
  \[
    P(x) \equiv \Pr\{\ell \le x\} = 1-\exp(-\alpha x^2) \enspace ,
  \]
  for $x\ge 0$. The probability density function of $\ell$ is therefore 
  given by 
  \[
     p(x) \equiv \frac{d}{dx}P(x) =
     2\alpha x\cdot\exp(-\alpha x^2) \enspace ,
  \]
  for $x\ge 0$.  Finally, we obtain $p_k$ as
  \begin{align*}
     p_k = \int_0^\infty p(\ell)\cdot f(\ell)\,\mathrm{d}\ell 
     = \frac{\pi\sqrt{3}}{9}
      \approx 0.6045997883  \enspace . & \qedhere
  \end{align*}
\end{proof}


\subsection{Analysis of $p_k$ for Odd $k$}
\seclabel{odd}

Next, we determine the values of $p_k$ for odd values of $k\ge 5$.
Although the strategy for doing this is the same as the even case,
the odd case turns out to be considerably more complicated; the value
of $p_k$ does, indeed depend on $k$.

\noindent
\begin{minipage}{\textwidth}
\begin{lem}\lemlabel{odd}
  For odd $k\ge 5$,
\[
p_k = 
2
\left(\begin{array}{l}
  \arctan\left(
     2\left(\cos\left(\frac{\pi }{k}\right)
       +\cos\left(\frac{3 \pi }{k}\right)\right)^2 / (\alpha\beta) 
  \right) \\
   + \arctan\left(
       4 \left(2 \cos\left(\frac{2 \pi }{k}\right)
       +\sin\left(\frac{2 \pi }{k}\right)^2\right)/(\alpha\beta) 
     \right)
  \end{array}
\right)
\cot\left(\frac{\pi }{k}\right) 
\beta
/
\left(\gamma \alpha\right)\enspace ,
\]
where
\[
\gamma =4+11 \cos\left(\frac{2 \pi }{k}\right)+\cos\left(\frac{6 \pi }{k}\right) \enspace ,
\]
\[
\alpha = 
\sqrt{\frac{\left(27 \cos\left(\frac{\pi }{k}\right)+17 \cos\left(\frac{3 \pi }{k}\right)+3 \cos\left(\frac{5 \pi }{k}\right)+\cos\left(\frac{7 \pi }{k}\right)\right) \csc\left(\frac{\pi }{k}\right)}{\gamma}} \enspace ,
\]
and
\[
\beta = \sqrt{\left(18 \sin\left(\frac{2 \pi }{k}\right)+18 \sin\left(\frac{4 \pi }{k}\right)+11 \sin\left(\frac{6 \pi }{k}\right)+\sin\left(\frac{8 \pi }{k}\right)+\sin\left(\frac{10 \pi }{k}\right)\right)} \enspace .
\]
\end{lem}
\end{minipage}

\begin{proof}
  The proof proceeds in the same manner as the proof of \lemref{even}.
  Let $u$ be an arbitrary vertex in a $\theta_k$-graph and let $w$ be
  a vertex that $u$ has chosen as a neighbour in one of its cones, $C$.
  Let $T$ be the open isosceles triangle defined by $C$ and a line through
  $w$ that is orthogonal to the axis of $C$. See \figref{mutual-odd}.
  Assume that the side of $T$ opposite $u$ has length $2\ell$.  Using a
  suitable rotation, we may assume that the axis of $C$ is horizontal and,
  by symmetry, we may assume that $w$ is on or above the axis of $C$.
 
  \begin{figure}
    \includegraphics{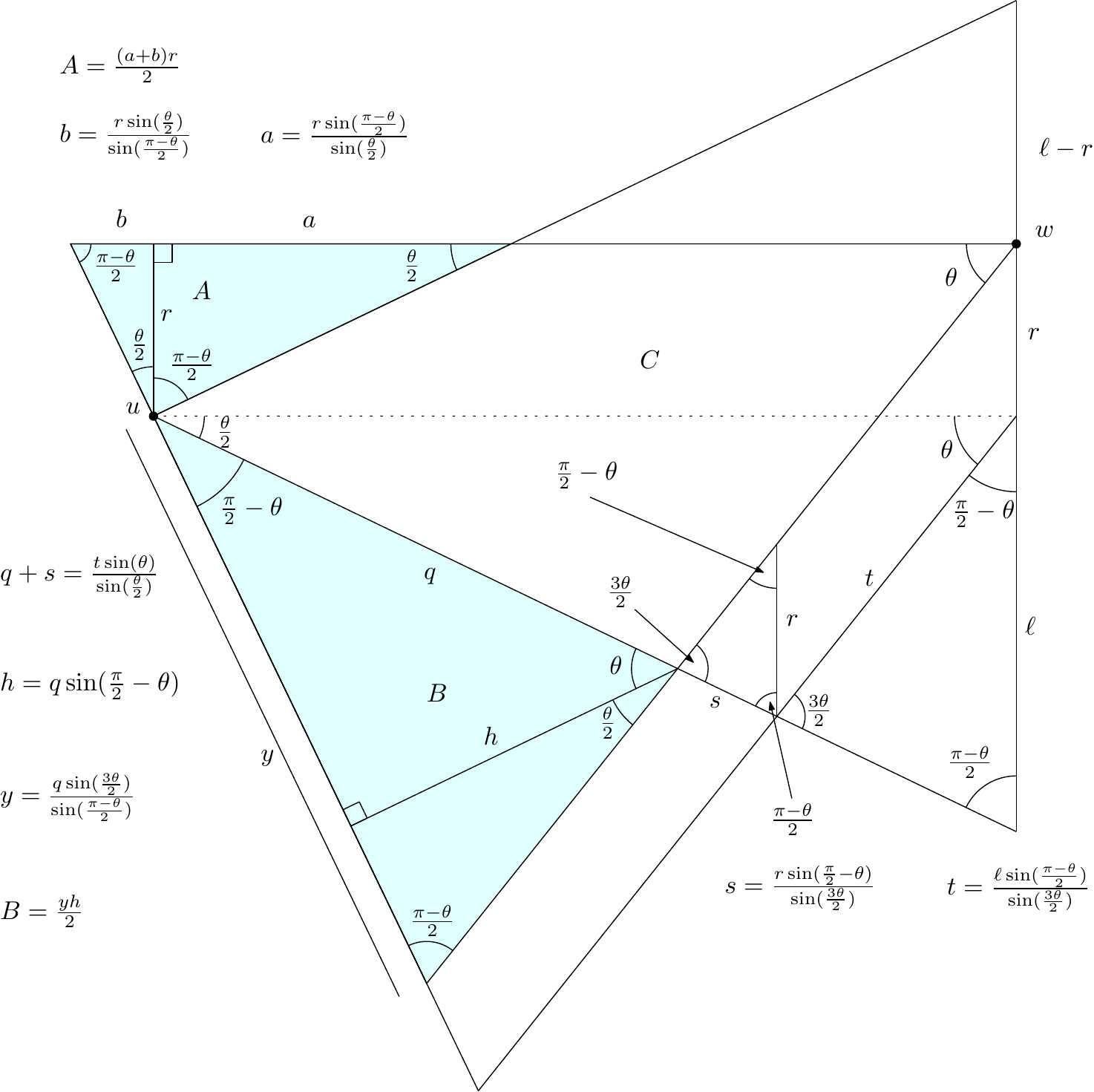}
    \caption{The derivation of the $\Pr\{\mbox{$uw$ is mutual}\mid \ell,r\}$
             for odd $k$.}
    \figlabel{mutual-odd}
  \end{figure}
  Under the preceding assumptions, $w$ is then uniformly distributed on
  a vertical segment of length $\ell$ whose endpoints are on the axis
  of $C$ and the upper boundary of $C$.  Suppose that the distance from
  $w$ to the axis of $C$ is $r$.  Then a straightforward, but tedious,
  calculation that mainly uses the law of sines shows that
  \[
      \Pr\{\mbox{$uw$ is mutual}\mid \ell, r\} = \exp(-A-B) \enspace ,
  \]
  where
  \[ 
      A = \frac{r^2}{\sin(2\pi/k)}
  \]
  and
  \[
      B = \frac{\cos(2\pi/k)}{2\sin(3\pi/k)\cos(\pi/k)}\cdot\left((\ell-r)\cos(2\pi/k)+\ell\right)^2
  \]
  This calculation is illustrated in \figref{mutual-odd} and the
  accompanying worksheet shows the simplifications that lead to the
  expressions for $A$ and $B$.  Integrating over $r$ gives us
  \begin{equation}
    f(\ell)\equiv \Pr\{\mbox{$uw$ is mutual}\mid \ell\} = 
      \int_0^\ell (1/\ell)\exp(-A-B)\,\mathrm{d}r . \eqlabel{odd-f}
  \end{equation}
  Like the corresponding integral in the proof of \lemref{even},
  \eqref{odd-f} has a closed-form that includes the Gauss error function.

  In order to remove the conditioning on $\ell$, we need the probability
  density function for $\ell$.  Proceeding as before, we have the
  cumulative distribution function
  \[
     P(x) \equiv \Pr\{\ell\le x\} 
          = 1 - \exp(x^2/\tan(\pi/k)) \enspace ,
  \]
  and the probability density function
  \[
    p(x)\equiv \frac{d}{dx}P(x) = 
     \left(\frac{2x}{\tan(\pi/k)}\right)
      \exp(x^2/\tan(\pi/k))
  \]
  Finally, we determine $p_k$ by integrating over $\ell$:
  \[
     p_k = \int_0^\infty p(\ell)\cdot f(\ell)\, \mathrm{d}{\ell} \enspace ,
  \]
  which (after introducing the variables $\alpha$, $\beta$, and $\gamma$)
  yields the expression for $p_k$ given in the statement of the lemma.
\end{proof}

\section{Points in a Unit Square}
\seclabel{iud}

Next we argue that results similar to Lemmas~\ref{lem:even} and
\ref{lem:odd}, albeit with lower-order error terms, hold for the
graph $\theta_k(S)$, where $S$ is a set of $n$ points independently
and uniformly distributed in the square $[0,\sqrt{n}]^2$ of area $n$.
Observe that, in this model, the probability that any particular region
$X\subseteq [0,\sqrt{n}]^2$ does not contain any points of $S$ is exactly
$(1-A/n)^n=\exp(-A)-O(A/n)$, where $A$ is the area of $X$.  This is
consistent with the Poisson model up to an additive error of $O(1/n)$.

The primary work in this section involves finding quantities that look
like those that appear in the previous section, but have an additive
lower-order error term.  To help manage these error terms, the
notation $x= y\pm a$ denotes that $x$ is some value in the interval
$[y-a,y+a]$.  We will abuse this notation slightly by writing equations
of the form $x\pm a = y\pm b$ when $[x-a,x+a]\subseteq[y-b,y+b]$.
Occasionally, we may also integrate expressions that use this notation.
In this case, we use the inequality $\int_a^b (x\pm c)\,\mathrm{d}x =
\int_a^b x \pm c(b-a)$.

We sometimes encounter expressions like $A/(1-x)$, with $0<x<1/2$ which
we bound by
\[
   A/(1-x) = A+O(Ax) \enspace .
\]

We also frequently encouter expressions like $(1-A/n)^{n-c}$, where $c$
is a constant and $A < n/2$.  We will always bound these as follows:
\begin{align*}
   (1-A/n)^{n-c} 
      & = \frac{(1-A/n)^n}{(1-A/n)^c} \\
      & = \frac{\exp(-A)-O(A/n)}{(1-A/n)^c} \\
      & = \frac{\exp(-A)-O(A/n)}{\sum_{i=0}^c \binom{c}{i}(-A/n)^i} \\
      & = \frac{\exp(-A)-O(A/n)}{1-O(A/n)} \\
      & \ge \exp(-A) - O(A/n) 
\end{align*}
and, similarly, 
\begin{align*}
   (1-A/n)^{n-c} 
      & = \frac{(1-A/n)^n}{(1-A/n)^c} \\
      & \le \frac{\exp(-A)}{(1-A/n)^c} \\
      & = \exp(-A)(1+O(A/n))  \enspace .
\end{align*}

\subsection{Expected Number of Edges}

In this section, we analyze the expected number of edges of $\theta_k(S)$.
For each point $u\in S$ and each $i\in\{1,\ldots,k\}$, let $e(u,i)$ be
the edge (if any) that $u$ chooses in its $i$th cone, $u+C_i$.  That is,
$e(u,i)$ is the edge $uw$ where $w\in u+C_i$ has the projection onto
the axis of $u+C_i$ that is smallest among all points in $S\cap u+C_i$.
We define the \emph{height} of the edge $uw=e(u,i)$ as the distance
between $u$ and the of the orthogonal projection of $w$ onto the axis
of $u+C_i$.

For our analysis, we partition $[0,\sqrt{n}]^2$ into a
\emph{core}, $C=[2t,\sqrt{n}-2t]^2$, where $t=\sqrt{ck\log n}$,
and a \emph{near-boundary}, $\bar{C}=[0,\sqrt{n}]^2\setminus C$
(see \figref{empty}).  The motivation for partitioning into a core and
near-boundary is that (1)~there are not many points in the near-boundary
and (2)~points in the core behave almost exactly like points in the
Poisson model.  The following Lemma shows, for example, that points in
the core nearly always have a neighbour in each of their cones.

\begin{figure}
  \begin{center}
    \includegraphics{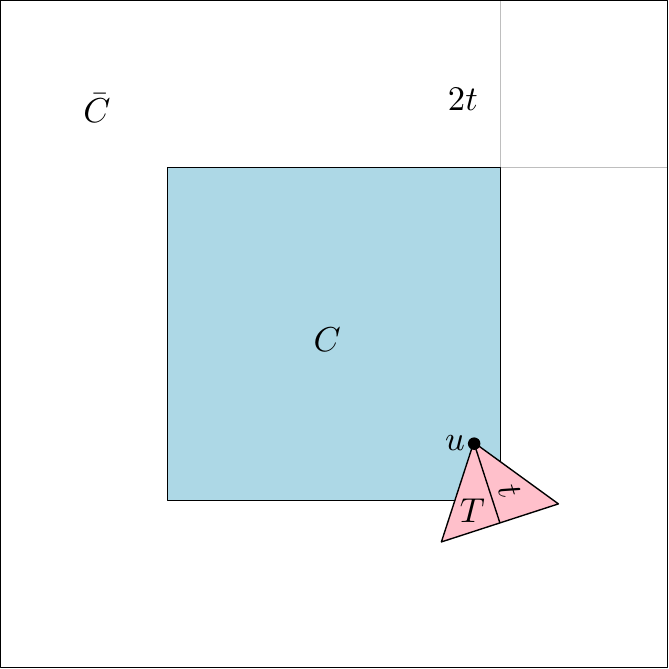}
  \end{center}
  \caption{The support, $[0,\sqrt{n}]^2$, is partitioned into a core, $C$,
    and a near-boundary $\bar{C}$.  A point $u$ in the core almost surely
    has a neighbour in every cone, otherwise $u$ is incident on a large
    triangle, $T\subset [0,\sqrt{n}]^2$, that is empty of points.}
  \figlabel{empty}
\end{figure}

\newcommand{\eui}{\mathcal{E}_{u,i}}
\newcommand{\neui}{\bar{\mathcal{E}}_{u,i}}

\begin{lem}\lemlabel{exists}
  For any $u\in S$ and any $i\in\{1,\ldots,k\}$, let $\eui$ denote
  the event ``$e(u,i)$ exists and has height at most $t$.''  Then
  $\Pr\{\eui\mid u\in C\} \ge 1-n^{-\Omega(c)}$.
\end{lem}

\begin{proof}
  Fix some location of $u\in C$ and draw an open isosceles triangle,
  $T$, with apex $u$, contained in $u+C_i$, whose internal angle at
  $u$ is $\theta=2\pi/k$ and whose height is $t$ (see \figref{empty}).
  Observe that, since $u$ is in the core, $T\subset [0,\sqrt{n}]^2$.
  Furthermore, the area of $T$
  is 
  \[ 
     A = t^2\tan(\pi/k)\in \Theta(t^2/k) = \Theta(c\log n) \enspace .
  \]
  For $T$ to be empty of points in $S$, the $n-1$ points of $S\setminus\{u\}$
  must all fall outside of $T$.  Thus, we have
  \begin{align*}
    \Pr\{T\cap S=\emptyset\mid u\in C\} 
       & = (1-A/n)^{n-1} \\
       & \le \exp(-A)(1+O(A/n)) \\
       & \le 2\exp(-A) & \text{(for sufficiently large $n$)}\\
       & = 2\exp(-\Theta(c\log n)) \\
       & = n^{-\Omega(c)} \enspace .
  \end{align*}
  If $\eui$ does not occur, this means that the event described above
  has occured.  Therefore, $\Pr\{\eui\mid u\in C\}\ge 1-n^{-\Omega(c)}$,
  as required.
\end{proof}

\newcommand{\mui}{\mathcal{M}_{u,i}}

Our next lemma shows that edges generated by points in the core have
essentially the same probability of being mutual as they do in the
Poisson model.

\begin{lem}\lemlabel{pk}
    For any $u\in S$ and any $i\in\{1,\ldots,k\}$, let $\mui$ denote the
    event ``$\eui$ and $e(u,i)$ is mutual.'' Then $\Pr\{\mui\mid u\in C\}
    = p_k\pm O((\log n)^2/n)$.
\end{lem}

\begin{proof}
Throughout this proof, all probabilities we compute are conditional on
$u\in C$, even when this is not explicitly stated.  The proof is basically
a reproving of Lemmas~\ref{lem:even} and \ref{lem:odd} that takes care to
deal with boundary effects.  Here, we will prove the case for even $k$
(\lemref{even}) only.  The case for odd $k$ (\lemref{odd}) can be done
the same way.

We first compute the probability conditional on $\eui$ and for fixed
values of $\ell$ and $r$ as described in the proof of \lemref{even}
and illustrated in \figref{mutual}.  The notations $\ell$, $r$, $T$,
and $T'$ all have the same meaning as in the proof of \lemref{even}.
The edge $e(u,i)$ is mutual if and only if the remaining $n-2$ points
in $S\setminus\{u,w\}$---which are already conditioned on not being in
$T$---also fall outside of $T'$.  The probability that this happens is
\begin{align}
   \Pr\{\mui\mid\eui,\,\ell,\, r\}
        & = \left(\frac{1-\area(T\cup T')/n}{1-\area(T)/n}\right)^{n-2} 
		\notag \\
        & = \left(\frac{1-\area(T'\setminus T)/n - \area(T)/n}
             {1-\area(T)/n}\right)^{n-2} \notag \\
        & = \left(1-\frac{\area(T'\setminus T)/n}{1-\area(T)/n}\right)^{n-2}
		\notag  \\
        & = \left(1-\frac{\alpha(r^2+(\ell-r)^2)/n)}
                       {1-\alpha\ell^2/n}\right)^{n-2} \notag \\
        & = \left(1-\frac{\alpha(r^2+(\ell-r)^2)
               (1+O(\alpha\ell^2/n))}{n} \right)^{n-2} \notag \\
        & = \exp(-\alpha(r^2+(\ell-r)^2)(1+O(\alpha\ell^2/n)))
         \pm O(\alpha\ell^2/n) \notag \\
        & = \exp(-\alpha(r^2+(\ell-r)^2))
            \cdot\exp(O(\alpha^2\ell^4/n)) \pm O(\alpha\ell^2/n) \notag \\
        & = \exp(-\alpha(r^2+(\ell-r)^2))(1+O(\alpha^2\ell^4/n)) 
            \pm O(\alpha\ell^2/n) \eqlabel{blech} \\
        & = \exp(-\alpha(r^2+(\ell-r)^2))
            \pm O((\alpha^2\ell^4+\alpha\ell^2)/n) \notag \\
        & = \exp(-\alpha(r^2+(\ell-r)^2))
            \pm O((\alpha^2 t^4+\alpha t^2)/n) \enspace . \notag
\end{align}
(Step \eqref{blech} follows from the inequality $1+(e-1)x \ge e^{x}$
for $0\le x\le 1$.)  We then remove the conditioning on $r$ by integrating:
\begin{align*}
   f'(\ell) \equiv & \Pr\{\mui\mid\eui,\,\ell\} \\
     & = \int_0^\ell(1/\ell)\left(\exp(-\alpha(r^2+(\ell-r)^2)) 
           \pm O((\alpha^2 t^4+\alpha t^2)/n)\right)\,\mathrm{d}r \\
     & = f(\ell) \pm O((\alpha^2 t^4+\alpha t^2)/n) \enspace ,
\end{align*}
where $f(\ell)$ is the same $f(\ell)$ defined in the proof of \lemref{even}.

To finish, we need the distribution function for $\ell$ conditional on
$\eui$.  The triangle $T$ has area $\alpha\ell^2$ so the probability that
it is empty of points of $S\setminus\{u\}$ is $(1-\alpha\ell^2/n)^{n-1}$.
Therefore, for $0\le x\le t$, we have the cumulative distribution function
\begin{align*}
   P(x) & \equiv \Pr\{\ell \le x\mid \eui\}  \\
      & = \frac{\Pr\{\eui\mbox{ and }\ell\le x\}}
               {\Pr\{\eui\}} \\
      & = \frac{\Pr\{\ell\le x\}}  
               {\Pr\{\eui\}} & \text{(since $0\le x\le t$, so $\ell \le x$ implies $\eui$)} \\
      & = \frac{1-(1-\alpha x^2/n)^{n-1}}{1-(1-\alpha t^2/n)^{n-1}} 
\end{align*}
From this we obtain the density function
\begin{align*}
  p'(x) & \equiv \frac{d}{dx}P(x) \\
        & = \frac{d}{dx}
             \frac{1-(1-\alpha x^2/n)^{n-1}}{1-(1-\alpha t^2/n)^{n-1}} \\
        & = \frac{2\alpha x(n-1)(1-\alpha x^2/n)^{n-2}}
                  {n(1-(1-\alpha t^2/n)^{n-1})} \\
        & = \frac{2\alpha x(n-1)(\exp(-\alpha x^2)\pm O(\alpha x^2/n))}
                  {n(1-n^{-\Omega(c)})} \\
        & = \frac{2\alpha x(n-1)(\exp(-\alpha x^2)\pm O(\alpha x^2/n))}
                  {n-n^{-\Omega(c)}} \\
        & = \frac{2\alpha xn(\exp(-\alpha x^2)\pm O(\alpha x^2/n))}
                  {n-n^{-\Omega(c)}}(1-1/n) \\
        & = \frac{2\alpha x(\exp(-\alpha x^2)\pm O(\alpha x^2/n))}
                  {1-n^{-\Omega(c)}}(1-1/n) \\
        & = 2\alpha x\left(\exp(-\alpha x^2)\pm O(\alpha x^2/n)\right)
                  (1+n^{-\Omega(c)})(1-1/n) \\
        & = 2\alpha x\left(\exp(-\alpha x^2)\pm O((1+\alpha x^2)/n)\right) \\
        & = 2\alpha x\exp(-\alpha x^2)\pm O((\alpha x+\alpha^2 x^3)/n) \\
        & = 2\alpha x\exp(-\alpha x^2)\pm O((\alpha t + \alpha^2 t^3)/n) \\
        & = p(x)\pm O(\alpha^2 t^3/n) ,
\end{align*}
where $p(x)$ is the same $p(x)$ that appears in the proof of \lemref{even}.
And now we have enough information to finish:
\begin{align*}
    \Pr\{\mui \mid u\in C\} 
         & = \Pr\{\eui\mid u\in C\}\cdot \Pr\{\mui\mid \eui,\,u\in C\} 
            \\ & \qquad {}+\Pr\{\neui\mid u\in C\}\cdot\Pr\{\mui\mid\neui,\,u\in C\} \\
         & \ge \Pr\{\eui\mid u\in C\}\cdot \Pr\{\mui\mid \eui,\,u\in C\} \\
         & \ge \left(1-n^{-\Omega(c)}\right)
                \Pr\{\mui\mid \eui,\,u\in C\} \\
         & = \left(1-n^{-\Omega(c)}\right)
                 \int_0^t p'(\ell)f'(\ell)\,\mathrm{d}\ell \\
         & = \left(1-n^{-\Omega(c)}\right)
            \int_0^t (p(\ell)\pm O(\alpha^2t^3/n))
           \cdot(f(\ell)\pm O((\alpha^2 t^4+\alpha t^2)/n))\,\mathrm{d}{\ell} \\
         & = \left(1-n^{-\Omega(c)}\right)\int_0^t p(\ell)f(\ell)
           \pm O(p(\ell)(\alpha^2t^4+\alpha t^2)/n + f(\ell)\alpha^2t^3/n + (\alpha^5t^8+\alpha^3t^5)/n^2)
            \,\mathrm{d}{\ell} \\
         & = \left(1-n^{-\Omega(c)}\right)\int_0^t p(\ell)f(\ell)
           \pm O((\alpha^2t^4+\alpha t^2)/n + (\alpha^5t^8+\alpha^3t^5)/n^2)
            \,\mathrm{d}{\ell} \\
        & = \frac{\pi\sqrt{3}}{9} \pm O((\alpha^2t^4+\alpha t^2)/n) \\
        & =  \frac{\pi\sqrt{3}}{9} \pm O((\log n)^2/n)
\end{align*}
and
\begin{align*}
    \Pr\{\mui \mid u\in C\} 
       & \le \Pr\{\eui\mid u\in C\}\cdot \Pr\{\mui\mid \eui\}
           + (1-\Pr\{\eui\}) \\
       & \le \Pr\{\eui\mid u\in C\}\cdot \Pr\{\mui\mid \eui\} 
           + n^{-\Omega(c)} \\
         & = \frac{\pi\sqrt{3}}{9} \pm O((\log n)^2/n) \enspace . \qedhere
\end{align*} 
\end{proof}

\begin{lem}
 Let $S$ be a set of $n$ points independently and uniformly distributed
 in $[0,1]^2$.  Then the expected number of edges of $\theta_k(S)$ is 
 $nd_k/2\pm O(k\sqrt{nk\log n})$.
\end{lem}

\begin{proof}
  In this proof, it will be helpful to distinguish between directed
  and undirected edges. Undirected edges are the edges $\theta_k(S)$
  that we have been considering throughout.  In contrast, $uw=e(u,i)$
  is a \emph{directed edge} from $u$ to $w$.  If $uw$ is mutual, then
  $wu=e(w,j)$ is a distinct directed edge from $w$ to $u$.  In this way,
  if we let $E$ denote the number of undirected edges, $D$ the number
  of directed edges, and $M$ the number of mutual directed edges of
  $\theta_k(S)$.   Then we have
  \[  
     E = D - M/2 \enspace . 
  \]
  Let $E_C$, $D_C$, and $M_C$ denote the same quantities but only
  counting those edges with at least one endpoint in the core. (See \figref{arrows}.)

  \begin{figure}
    \begin{center}
      \includegraphics{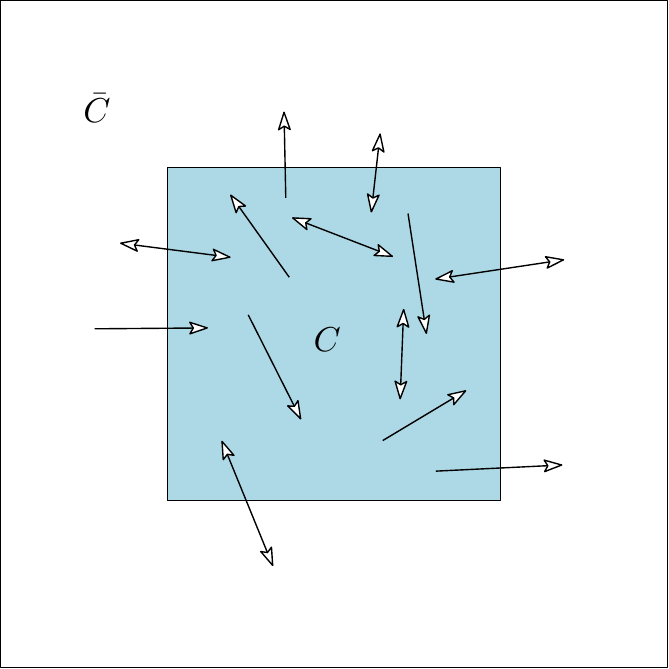}
    \end{center}
    \caption{The quantities $E_C$, $D_C$, and $M_C$: $E_C$ counts the
      number of segments,  $D_C$ counts the number of arrowheads, and $M_C$
      counts the number of segments with two arrowheads.}
    \figlabel{arrows}
  \end{figure}
  We begin with a lower bound:
  \begin{align*}
   \E[E] & \ge \E[E_C] \\
         & = \E[D_C] - \E[M_C]/2 \\
         & = \sum_{u\in S}\sum_{i=1}^k\Pr\{u\in C\}
              \left(\Pr\{\eui\mid u\in C\}
                - \Pr\{\mui\mid u\in C\}/2 \right) \\
        & = k(\sqrt{n}-4t)^2
              \left(\Pr\{\eui\mid u\in C\}
                - \Pr\{\mui\mid u\in C\}/2 \right) \\
        & \ge k(\sqrt{n}-4t)^2
              \left(1-n^{-\Omega(c)} - p_k/2 - O((\log n)^\frac52/n)\right) \\
        & \ge (kn - 8t\sqrt{n})
              \left(1-n^{-\Omega(c)} - p_k/2 - O((\log n)^\frac52/n)\right) \\
        & = kn(1-p_k/2) - O(t\sqrt{n}) \\
        & = nd_k/2 - O(\sqrt{nk\log n})
  \end{align*} 
  For the upper-bound we proceed as follows:
  \begin{align*}
     \E[E] & \le \E[E_C] + \E[k\cdot|S\cap\bar{C}|] \\
           & = \E[E_C] + k(8t\sqrt{n}-16t^2) \\
           & \le \E[E_C] + 4k\sqrt{nck\log n} \\
           & \le nd_k/2 + O((\log n)^\frac52) + 4k\sqrt{nck\log n} \\
           & = nd_k/2 + O(k\sqrt{nk\log n})
  \end{align*}
  where the last step follows from a calculation similar to that
  done in the lower-bound.
%
%
%
\end{proof}

\subsection{Concentration}

Next, we show that the number of edges in this model is tightly
concentrated about its expected value.  We begin with the following
result, which follows immediately from Hoeffding's Inequality
\cite{boucheron.lugosi.ea:concentration}, and shows that the number of
points in the core is highly concentrated around its expected value:
\begin{lem}\lemlabel{coresize}
  $\Pr\left\{\left||S\cap C|-(\sqrt{n}-4t)^2\right| \ge \sqrt{ck\log n}\right\}
  \in n^{-\Omega(c)}$.
\end{lem}

To prove our concentration bound, we make use of a versatile concentration
inequality due to McDiarmid \cite[Lemma~1.2]{mcdiarmid:on}:

\begin{lem}[McDiarmid's Inequality]
Let $X_1,\ldots,X_n$ be independent all taking values in the set
$\mathcal{X}$ and let $f:\mathcal{X}^n\mapsto\R$ be a function such that
\[
    |f(x_1,\ldots,x_i,\ldots,x_n) - f(x_1,\ldots,x_i',\ldots,x_n)|
    \le k \enspace ,
\]
for some $k>0$ and all $x_1,\ldots,x_n,x_i'\in \mathcal{X}$
and all $i\in\{1,\ldots,n\}$.
Then, for all $\epsilon > 0$,
\[
    \Pr\left\{\left|f(X_1,\ldots,X_n) - \E[f(X_1,\ldots,X_n)]\right| \ge \epsilon\right\}
       \le 2\exp(-2\epsilon^2/(nk^2)) \enspace .
\]
\end{lem}

\begin{lem}
 Let $S$ be a set of $n$ points independently and uniformly distributed
 in $[0,\sqrt{n}]^2$ and let $m_k$ denote the number of edges in $\theta_k(S)$.
 Then $\Pr\{|m_k-\E[m_k]| \ge k\sqrt{cn\log n}\} \le n^{-\Omega(c)}$.
\end{lem}

\begin{proof}
Let $Q$ be a set of $k$ points chosen such that, for any point
$u\in[0,\sqrt{n}]^2$, each of $u$'s $\theta_k$-cones contains exactly
one point in $Q$. ($Q$ could be, for example, the vertices of a
sufficiently large regular $k$-gon. See \figref{q}.)  We begin by studying
$\theta_k(S\cup Q)$.  This graph is somewhat more nicely behaved since
each vertex in $S$ defines exactly $k$ directed edges.

\begin{figure}
  \begin{center}
    \includegraphics{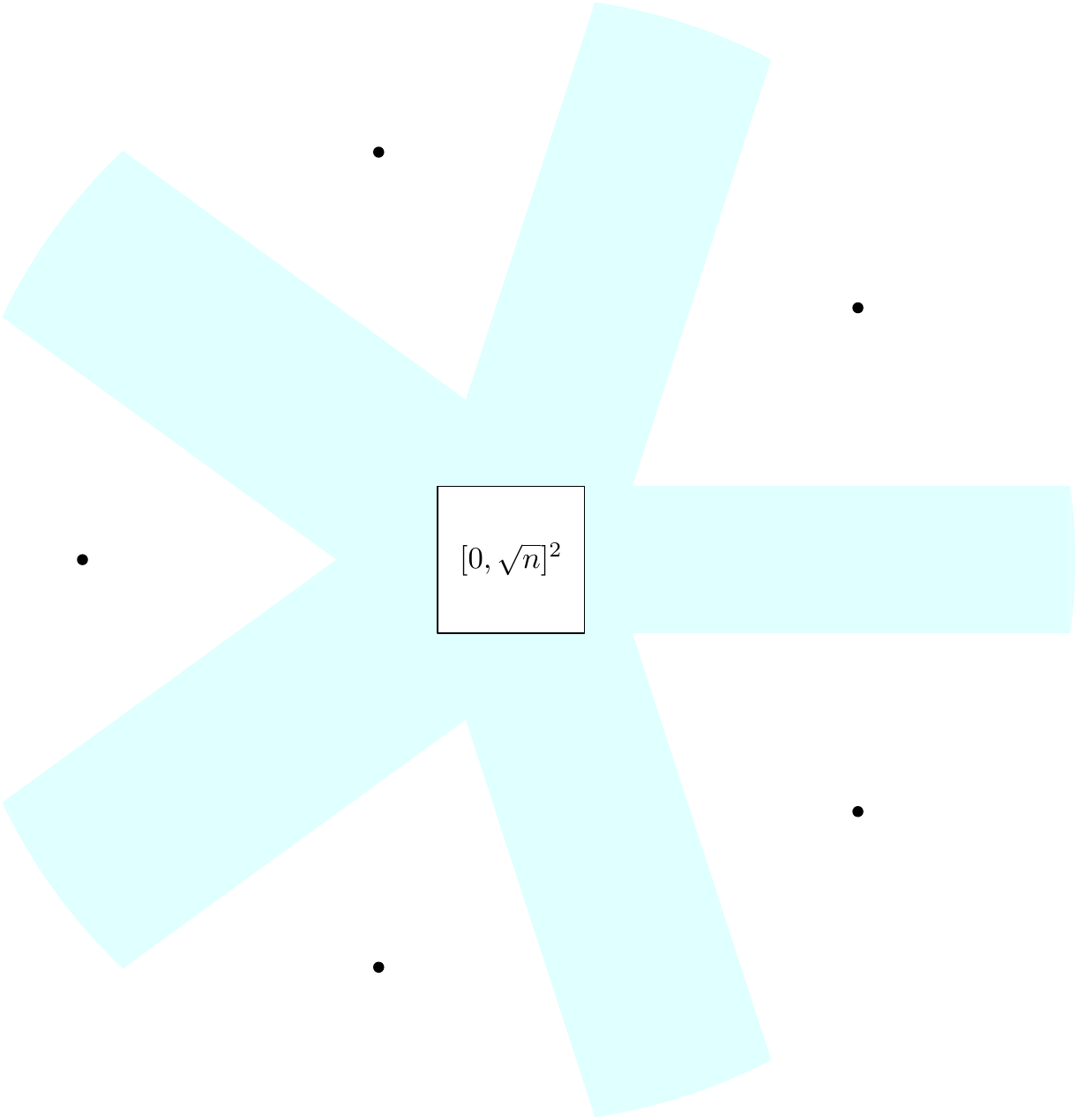}
  \end{center}
  \caption{The set $Q$ ensures that every vertex in $S$ has at least
    one point in each of its $\theta_k$-cones.  This figure shows a
    possible set $Q$ for $k=5$.}
  \figlabel{q}
\end{figure}

For a set $S=\{u_1,\ldots,u_n\}\subset [0,\sqrt{n}]^2$, let
$f(u_1,\ldots,u_n)$ be the function that counts the number of edges
in $\theta_k(S\cup Q)$.  Observe that, for any $u_1,\ldots,u_n,u_i'\in
[0,1]^2$, we have that
\[
   |f(u_1,\ldots,u_i,\ldots,u_n)-f(u_1,\ldots,u_i',\ldots,u_n)| \le k \enspace .
\]
That is, moving $u_i$ to location $u_i'$ can be though of as removing
$u_i$, resulting in the loss of at most $k$ non-mutual edges emanating
from $u_i$, followed by adding $u_i'$ resulting in the creation of at
most $k$ non-mutual edges emanating from $u_i'$.  Letting $m_k'$ denote
the number of edges in $\theta_k(S\cup Q)$, we immediately obtain,
from McDiarmid's Inequality,
\[
   \Pr\left\{\left|m_k'-\E[m_k']\right| \ge k\sqrt{cn\log n}\right\} 
       \le 2\exp(-2c\log n) 
       \in n^{-\Omega(c)} \enspace .
\]
To extend this result to $m_k$, the number of edges of $\theta_k(S)$,
we study $\Pr\{|m_k' - m_k|\ge k\sqrt{cn\log n}\}$.  If $|m_k' - m_k|\ge
k\sqrt{cn\log n}$ then the number of points of $S$ in the near-boundary
$\bar{C}$ exceeds $\sqrt{cn\log n}$ or $\theta_k(S\cup Q)$ contains
edges that join points in $Q$ to points in the core, $C$.
\lemref{coresize} shows that the probability of the former event is at most
$n^{-\Omega(c)}$, while \lemref{exists} shows that the probability of
the latter event is at most $n^{-\Omega(c)}$
\end{proof}

\section{Discussion}
\seclabel{summary}

We have given exact closed-form expressions for the average degree of
$\theta_k$-graphs for all $k\ge 4$.  It is easily shown that $\theta_k$
graphs with $k=1,2,3$ are not spanners, so the cases $k\ge 4$ are
the most important.  We have also shown that the number of edges in
a $\theta_k$-graph of $n$ points uniformly distributed in a square is
highly concentrated around its expected value.  These results can be
used to inform practitioners about the optimal choice of $k$ to use in
a particular application.  \tabref{values} gives the numerical values
of $d_k$, $p_k$, and $d_k/k$, for $k\in\{4,\ldots,20\}$.

\begin{table}
   \begin{center}
     \begin{tabular}{r|r|r|r|r|}
       \multicolumn{1}{c|}{$k$} & \multicolumn{1}{c|}{$p_k$} & \multicolumn{1}{c|}{$d_k$} & \multicolumn{1}{c|}{$d_k/k$} & \multicolumn{1}{c|}{$\hat{d}_k$} \\ \hline
         4 & 0.60459978 &  5.58160084 & 1.39540021 &  5.58190420  \\
         5 & 0.70168463 &  6.49157685 & 1.29831537 &  6.49280635  \\
         6 & 0.60459978 &  8.37240127 & 1.39540021 &  8.37587443  \\
         7 & 0.66778040 &  9.32553719 & 1.33221959 &  9.32743266  \\
         8 & 0.60459978 & 11.16320170 & 1.39540021 & 11.16338275  \\
         9 & 0.65740932 & 12.08331611 & 1.34259067 & 12.08319579  \\
        10 & 0.60459978 & 13.95400212 & 1.39540021 & 13.95657530  \\
        11 & 0.65259895 & 14.82141149 & 1.34740104 & 14.82372518  \\
        12 & 0.60459978 & 16.74480254 & 1.39540021 & 16.74624814  \\
        13 & 0.64993659 & 17.55082423 & 1.35006340 & 17.55520514  \\
        14 & 0.60459978 & 19.53560297 & 1.39540021 & 19.53693380  \\
        15 & 0.64830027 & 20.27549589 & 1.35169972 & 20.27565544  \\
        16 & 0.60459978 & 22.32640339 & 1.39540021 & 22.32956115  \\
        17 & 0.64722020 & 22.99725644 & 1.35277979 & 22.99979580  \\
        18 & 0.60459978 & 25.11720381 & 1.39540021 & 25.12062469  \\
        19 & 0.64646902 & 25.71708858 & 1.35353097 & 25.71956002  \\
        20 & 0.60459978 & 27.90800424 & 1.39540021 & 27.91343314  
     \end{tabular}
   \end{center}
   \caption{Numeric values of $p_k$ and $d_k$.}
   \tablabel{values}
\end{table}

Some of our results involved extensive symbolic manipulations that were
done with the help of Mathematica.  There is certainly the possibility
of mistakes, either in the software or by its users.  In order to help
validate our analytical results, the final column in \tabref{values} also
shows the results of the following experiment: 4,000 points were generated
uniformly in the unit square $[0,1]^2$ and their $\theta_k$-graph was
computed.  The average degree of points in the square $[1/3,2/3]^2$ was
then computed.  The final column of \tabref{values} shows the average of
these values taken over all 2,000 repetitions of this experiment.  In all
cases, it agrees with our theoretical results up to the first 3 digits.

Yao graphs \cite{flinchbaugh.jones:strong,yao:on} are closely related to
theta-graphs.  In the Yao graph, $Y_k(S)$, each vertex $u$ is connected by
an edge to the closest point in each of its  theta cones, where closest
is defined in terms of Euclidean distance.  The same strategy we use to
determine the average degree of $\theta_k(S)$ can be applied to $Y_k(S)$.
Unfortunately, this fails to give closed form answers.  In particular,
when trying to repeat the proof of \lemref{even}, we obtain the formula
\[
  p'_k = \int_0^\infty \theta t\exp(-\theta t^2/2) \int_0^t \exp(-\theta t^2/2+t^2\sin(\gamma)\sin(\theta-\gamma)/\sin(\theta))\,\mathrm d\gamma\, \mathrm dt \enspace ,
\]
for which we are unable to obtain a closed form (here $\theta=2\pi/k$).
Nevertheless, one can evaluate this integral numerically to obtain
estimates of the mutual edge probability and average degree in Yao graphs.

\section*{Acknowledgement}

The authors of this paper are partly funded by NSERC and CFI.

\bibliographystyle{abbrv}
\bibliography{avgtheta}

\begin{thebibliography}{10}

\bibitem{abramowitz.stegun:handbook}
M.~Abramowitz and I.~A. Stegun, editors.
\newblock {\em Handbook of Mathematical Functions}.
\newblock Dover Publications, New York, 9th edition, 1972.

\bibitem{alzoubi.li.ea:geometric}
K.~M. Alzoubi, X.-Y. Li, Y.~Wang, P.-J. Wan, and O.~Frieder.
\newblock Geometric spanners for wireless ad hoc networks.
\newblock {\em IEEE Trans. Parallel Distrib. Syst.}, 14(4):408--421, 2003.

\bibitem{arkin.anta.ea:probabilistic}
E.~M. Arkin, A.~F. Anta, J.~S.~B. Mitchell, and M.~A. Mosteiro.
\newblock Probabilistic bounds on the length of a longest edge in delaunay
  graphs of random points in d-dimensions.
\newblock In {\em Proceedings of the 23rd Annual Canadian Conference on
  Computational Geometry (CCCG 2011)}, 2011.

\bibitem{avis.beresford-smith.ea:unoriented}
D.~Avis, B.~Beresford-Smith, L.~Devroye, H.~A. ElGindy, E.~Gu{\'e}vremont,
  F.~Hurtado, and B.~Zhu.
\newblock Unoriented $\theta$-maxima in the plane: Complexity and algorithms.
\newblock {\em SIAM Journal on Computing}, 28(1):278--296, 1998.

\bibitem{barba.bose.ea:on}
L.~Barba, P.~Bose, J.-L. de~Carufel, A.~van Renssen, and S.~Verdonschot.
\newblock On the stretch factor of the theta-4 graph.
\newblock arXiv:1303.5473, 2012.

\bibitem{bern.eppstein.ea:expected}
M.~Bern, D.~Eppstein, and F.~Yao.
\newblock The expected extremes in a {D}elaunay triangulation.
\newblock {\em International Journal of Computational Geometry and
  Applications}, 1:79--91, 1991.

\bibitem{bonichon.gavoille.ea:connections}
N.~Bonichon, C.~Gavoille, N.~Hanusse, and D.~Ilcinkas.
\newblock Connections between theta-graphs, delaunay triangulations, and
  orthogonal surfaces.
\newblock In {\em WG}, volume 6410 of {\em Lecture Notes in Computer Science},
  pages 266--278. Springer, 2010.

\bibitem{bose.morin.ea:theta-5}
P.~Bose, P.~Morin, A.~van Renssen, and S.~Verdonschot.
\newblock The theta-5 graph is a spanner.
\newblock arXiv:1212.0570, 2012.

\bibitem{boucheron.lugosi.ea:concentration}
S.~Boucheron, G.~Lugosi, and O.~Bousquet.
\newblock Concentration inequalities.
\newblock In {\em Advanced Lectures on Machine Learning}, volume 3176 of {\em
  Lecture Notes in Computer Science}, pages 208--240. Springer, 2004.

\bibitem{clarkson:approximation}
K.~L. Clarkson.
\newblock Approximation algorithms for shortest path motion planning (extended
  abstract).
\newblock In A.~V. Aho, editor, {\em Proceedings of the 19th Annual ACM
  Symposium on Theory of Computing (STOC'87)}, pages 56--65. ACM, 1987.

\bibitem{daley.vere-jones:introduction}
D.~J. Daley and D.~Vere-Jones.
\newblock {\em An Introduction to the Theory of Point Processes. Volume I:
  Elementary Theory and Methods}.
\newblock Probability and Its Applications. Springer, 2003.

\bibitem{devroye:expected}
L.~Devroye.
\newblock The expected size of some graphs in computational geometry.
\newblock {\em Computers and Mathematics with Applications}, 15:53--64, 1988.

\bibitem{devroye.gudmundsson.ea:on}
L.~Devroye, J.~Gudmundsson, and P.~Morin.
\newblock On the expected maximum degree of {G}abriel and {Y}ao graphs.
\newblock {\em Advances in Applied Probability}, 41(4):1123--1140, 2009.

\bibitem{fischer.lukovszki.ea:geometric}
M.~Fischer, T.~Lukovszki, and M.~Ziegler.
\newblock Geometric searching in walkthrough animations with weak spanners in
  real time.
\newblock In G.~Bilardi, G.~F. Italiano, A.~Pietracaprina, and G.~Pucci,
  editors, {\em Proceedings of the 6th Annual European Symposium on Algorithms
  (ESA'98)}, volume 1461 of {\em Lecture Notes in Computer Science}, pages
  163--174. Springer, 1998.

\bibitem{flinchbaugh.jones:strong}
B.~E. Flinchbaugh and L.~K. Jones.
\newblock Strong connectivity in directional nearest-neighbor graphs.
\newblock {\em SIAM Journal on Algebraic Discrete Methods}, 2(4), 1981.

\bibitem{keil:approximating}
J.~M. Keil.
\newblock Approximating the complete euclidean graph.
\newblock In R.~G. Karlsson and A.~Lingas, editors, {\em Proceedings of the 1st
  Scandinavian Workshop on Algorithm Theory (SWAT'88)}, volume 318 of {\em
  Lecture Notes in Computer Science}, pages 208--213, 1988.

\bibitem{keil.gutwin:classes}
J.~M. Keil and C.~A. Gutwin.
\newblock Classes of graphs which approximate the complete euclidean graph.
\newblock {\em Discrete {\&} Computational Geometry}, 7:13--28, 1992.

\bibitem{mcdiarmid:on}
C.~McDiarmid.
\newblock On the method of bounded differences.
\newblock In J.~Siemons, editor, {\em Surveys in Combinatorics}, volume 141 of
  {\em London Mathematical Society Lecture Note Series}, pages 148--188.
  Cambridge University Press, 1989.

\bibitem{meester.roy:continuum}
R.~Meester and R.~Roy.
\newblock {\em Continuum percolation}, volume 119 of {\em Cambridge Tracts in
  Mathematics}.
\newblock Cambridge University Press, Cambridge, 1996.

\bibitem{penrose:random}
M.~Penrose.
\newblock {\em Random Geometric Graphs}.
\newblock Oxford University Press, 2003.

\bibitem{penrose.yukich:central}
M.~D. Penrose and J.~E. Yukich.
\newblock Central limit theorems for some graphs in computational geometry.
\newblock {\em The Annals of Applied Probability}, 11(4):1005--1041, 2001.

\bibitem{penrose.yukich:weak}
M.~D. Penrose and J.~E. Yukich.
\newblock Weak laws of large numbers in geometric probability.
\newblock {\em The Annals of Applied Probability}, 13(1):277--303, 2003.

\bibitem{ruppert.seidel:approximating}
J.~Ruppert and R.~Seidel.
\newblock Approximating the {$d$}-dimensional complete {E}uclidean graph.
\newblock In {\em Proceedings of the 3rd Canadian Conference on Computational
  Geometry}, pages 207--210, 1991.

\bibitem{yao:on}
A.~C.-C. Yao.
\newblock On constructing minimum spanning trees in k-dimensional spaces and
  related problems.
\newblock {\em SIAM Journal on Computing}, 11(4):721--736, 1982.

\end{thebibliography}

\end{document}